\documentclass[letterpaper]{scrartcl}

\usepackage[amsthm]{pamas}

\usepackage{eucal}
\usepackage{nicefrac}

\title{The Impossibility of Extending Random Dictatorship to Weak Preferences}

\author{
Florian Brandl\\
TU M\"unchen\\
Germany\\
\texttt{\small brandlfl@in.tum.de}
\and
Felix Brandt\\
TU M\"unchen\\
Germany\\
\texttt{\small brandtf@in.tum.de}
\and
Warut Suksompong\\
Stanford University\\
USA\\
\texttt{\small warut@cs.stanford.edu}
}

\date{}

\newcommand{\pref}{\succcurlyeq}
\newcommand{\sPref}{\succ}
\newcommand{\indiff}{\sim}

\newcommand{\rd}[1][]{\ifthenelse{\equal{#1}{}}{\mathit{RD}}{\mathit{RD}(#1)}}
\newcommand{\rsd}[1][]{\ifthenelse{\equal{#1}{}}{\mathit{RSD}}{\mathit{RSD}(#1)}}
\newcommand{\sd}[1][]{\ifthenelse{\equal{#1}{}}{\mathit{SD}}{\mathit{SD}(#1)}}

\begin{document}

\maketitle

\begin{abstract}
Random dictatorship has been characterized as the only social decision scheme that satisfies efficiency and strategyproofness 
when individual preferences are strict. We show that no extension of random dictatorship to weak preferences satisfies these properties, even when significantly weakening the required degree of strategyproofness.
\end{abstract}


{\bf JEL Classifications Codes:} C6, D7, D8

\section{Introduction}

One of the most celebrated results in microeconomic theory is the Gibbard-Satterthwaite theorem \citep{Gibb73a,Satt75a}, which states that every strategyproof and Pareto-optimal social choice function is a dictatorship. However, the theorem crucially relies on the assumption that outcomes are deterministic. \citet{Gibb77a} later considered social decision schemes, \ie social choice functions that return lotteries over the alternatives, and showed that the class of strategyproof and \emph{ex post}  efficient functions extends to all random dictatorships. This class contains a unique rule that treats all agents equally: the uniform random dictatorship, henceforth \emph{random dictatorship ($\rd$)}, where an agent is chosen uniformly at random and his favorite alternative is implemented as the social choice. Gibbard's notion of strategyproofness is based on \emph{stochastic dominance} and prescribes that no voter can obtain more utility by misrepresenting his preferences no matter what his utility function is (as long as it is consistent with his ordinal preferences). Another implicit assumption in Gibbard's theorem is the anti-symmetry of individual preferences. Characterizations of strategyproof social decision schemes for the case when agents are allowed to express indifference have also been explored. In the context of cardinal decision schemes, \citet{DPS07a} characterize $\rd$ for the domain in which each agent has a unique top choice. For arbitrary indifferences, \citet{Hyll80a} and \citet{Nand12a} show that the only reasonable strategyproof social decision schemes are weak random dictatorships. We refer to \citet{Nand12a} for a discussion of these results. Perhaps the best-known generalization of $\rd$ to weak preferences is \emph{random serial dictatorship ($\rsd$)} where a permutation of agents is chosen uniformly at random and agents narrow down the set of alternatives in that order to their most preferred alternatives among the remaining alternatives. $\rsd$ is also \emph{ex post} efficient and strategyproof with respect to stochastic dominance. However, in contrast to $\rd$ it is not efficient with respect to stochastic dominance, \ie there might be a lottery that yields more expected utility for all agents. This failure of efficiency was first observed by \citet{BoMo01a} in the context of random assignment. We show that this is not a weakness specific to $\rsd$ but in fact all fair generalizations of $\rd$ violate either efficiency or strategyproofness, even when significantly weakening the required degree of strategyproofness.\footnote{For example, this also explains why another strategyproof extension of $\rd$ to weak preferences, the maximal recursive rule \citep{Aziz13b}, violates efficiency.}

\section{Preliminaries}

Let $N=\{1,\dots,n\}$ be a set of agents with preferences over a finite set $A$ with $|A| = m$. 
The preferences of agent~$i \in N$ are represented by a complete and transitive \emph{preference relation}~${\pref}_i\subseteq A\times A$. 
The set of all preference relations will be denoted by~$\mathcal{R}$.
In accordance with conventional notation, we write~$\sPref_i$ for the strict part of~$\pref_i$, \ie~$a \sPref_i b$ if~$a \pref_i b$ but not~$b \pref_i a$ and~$\indiff_i$ for the indifference part of~$\pref_i$, \ie~$a \indiff_i b$ if~$a \pref_i b$ and~$b \pref_i a$. 
We will compactly represent a preference relation as a comma-separated list with all alternatives among which an agent is indifferent placed in a set. 
For example $a \sPref_i b \indiff_i c$ will be written as 
${\pref}_i \colon a, \{b,c\}$.
A preference relation $\pref_i$ is \emph{strict} if $x \sPref y$ or~$y\sPref x$ for all distinct alternatives $x,y$.
A \emph{preference profile} $R = ({\pref}_1,\dots, {\pref}_n)$ is an $n$-tuple containing a preference relation $\pref_i$ for each agent $i \in N$. 
The set of all preference profiles is thus given by $\mathcal{R}^n$. 
By $R_{-i}$ we denote the preference profile obtained from $R$ by removing the preference relation of agent $i$, \ie $R_{-i} = R\setminus\{(i,{\pref}_i)\}\text.$

Let furthermore $\Delta(A)$ denote the set of all \emph{lotteries} (or \emph{probability distributions}) over $A$ and, for a given lottery $p\in\Delta(A)$,
$p(x)$ denote the probability that $p$ assigns to alternative~$x$. Lotteries will be written as convex combinations of alternatives, \eg $\nicefrac{1}{2}\,a + \nicefrac{1}{2}\,b$ denotes the lottery $p$ with $p(a)=p(b)=\nicefrac{1}{2}$.
is the set of all alternatives to which $p$ assigns positive probability.

Our central object of study are \emph{social decision schemes}, \ie functions that map the individual preferences of the agents to a lottery over alternatives. Formally,
a social decision scheme (SDS) is a function $f \colon \mathcal{R}^n \rightarrow \Delta(A)$.
A minimal fairness condition for SDSs is \emph{anonymity}, which requires that $f(R)=f(R')$ for all $R,R'\in \mathcal{R}^n$ and permutations $\pi\colon N\rightarrow N$ such that ${\pref}'_i={\pref}_{\pi(i)}$ for all $i\in N$. 
Another fairness requirement is \emph{neutrality}. 
For a permutation of alternatives $\sigma$ and a preference relation ${\pref}_i$, $\sigma(x) \pref_i^{\sigma} \sigma(y)$ if and only if $x \pref_i y$. 
Then, an SDS $f$ is \emph{neutral} if $f(R)(x) = f(R^{\sigma})(\sigma(x))$ for all $R\in\mathcal{R}^n$, $x \in A$, and permutations $\sigma\colon A\rightarrow A$.

Two well-studied SDSs are \textit{Random Dictatorship ($\rd$)} and \emph{Random Serial Dictatorship (RSD)}. $\rd$ is defined when all agents have a unique favorite alternative. This includes the domain of strict preferences as a subclass. The lottery returned by $\rd$ is obtained by choosing an agent uniformly at random and returning that agent's favorite alternative. 
$\rsd$ is an extension of $\rd$ to the full domain of preferences. $\rsd$ operates by first choosing a permutation of the agents uniformly at random. Starting with the set of all alternatives, it then asks each agent in the order of the permutation to choose his favorite alternative(s) among the remaining alternatives. If more than one alternative remains after taking the preferences of all agents into account, $\rsd$ uniformly randomizes over those alternatives.
Formally, we obtain the following recursive definition.
\begin{eqnarray*}
	\rsd(R,X)= 
	\begin{cases}
		\sum\limits_{x \in X} \frac{1}{|X|} ~x &\text{if }{R}=\emptyset\text{,}\\[1em]
		\sum\limits_{i=1}^{|R|} \frac{1}{|R|} ~\rsd(R_{-i},\max_{\pref_i}(X)) &\text{otherwise,}
	\end{cases}
\end{eqnarray*}
and $\rsd(R)=\rsd(R,A)$.
The formal definition of $\rd$ is a special case of the above definition of $\rsd$.
In contrast to deterministic dictatorships, $\rsd$ is anonymous and is frequently used in subdomains of social choice that are concerned with the fair assignment of objets to agents \citep[see, \eg][]{BCKM12a}.

\section{Efficiency and Strategyproofness}

In order to reason about the outcomes of SDSs, we need to make assumptions on how agents compare lotteries. A common way to extend preferences over alternatives to preferences over lotteries is \emph{stochastic dominance ($\sd$)}. A lottery $\sd$-dominates another if, for every alternative $x$, the former is at least as likely to yield an alternative at least as good as $x$ as the latter. Formally, 
\[p \pref_i^{\sd} q  \text{ iff for all }x\in A,
\sum_{y \colon y \pref_i x} p(y) \geq \sum_{y \colon y \pref_i x} q(y)\text{.}
\]
It is well-known that $p \pref_i^{\sd} q$ if and only if the expected utility for $p$ is at least as large as that for $q$ for every von Neumann-Morgenstern function consistent with $\pref_i$.

Thus, for the preference relation ${\pref}_i\colon a,b,c$, we for example have that 
\[
(\nicefrac{2}{3}\, a + \nicefrac{1}{3}\, c) \sPref_i^{\sd} (\nicefrac{1}{3}\, a + \nicefrac{1}{3}\, b + \nicefrac{1}{3}\, c),
\]
while $\nicefrac{2}{3}\, a + \nicefrac{1}{3}\, c$ and $b$ are incomparable.

In this section, we define the notions of efficiency and strategyproofness considered in this paper.
The two notions of efficiency defined below are generalizations of Pareto-optimality in non-probabilistic social choice.
An alternative is \emph{Pareto-dominated} if there exists another alternative such that all agents weakly prefer the latter to the former with a strict preference for at least one agent.
An SDS is \emph{ex post efficient} if it assigns probability zero to all Pareto-dominated alternatives \citep[see \eg][]{Gibb77a,BMS05a}.

Second, we define efficiency with respect to stochastic dominance. 
A lottery $p$ is \emph{$\sd$-efficient} if there is no other lottery $q$ that is weakly $\sd$-preferred by all agents with a strict preference for at least one agent, \ie $q\pref_i^{\sd} p$ for all $i\in N$ and $q\sPref_i^{\sd} p$ for some $i\in N$. 
It is well-known that $\sd$-efficiency is stronger than \emph{ex post} efficiency.
An SDS is $\sd$-efficient if it returns an $\sd$-efficient lottery for every preference profile \citep[see, \eg][]{BoMo01a,ABB13d,ABB14b}.

For better illustration, consider $A=\{a,b,c,d\}$ and the preference profile $R=({\pref}_1,\dots,{\pref}_4)$,
	\begin{align*}
		\pref_1\colon &\{a,c\},b,d&
		\pref_2\colon &\{b,d\},a,c&
		\pref_3\colon &a,d,b,c&
		\pref_4\colon &b,c,a,d
	\end{align*}
	Observe that no alternative is Pareto-dominated, \ie for instance the uniform lottery $\nicefrac{1}{4}\, a+\nicefrac{1}{4}\, b+\nicefrac{1}{4}\, c+\nicefrac{1}{4}\, d$ is \emph{ex post} efficient. 
	On the other hand, the uniform lottery is not $\sd$-efficient as all agents strictly $\sd$-prefer $\nicefrac{1}{2}\, a+\nicefrac{1}{2}\, b$.

Strategyproofness prescribes that no agent can obtain a more preferred outcome by misrepresenting his preferences. 
There are two notions of strategyproofness associated with stochastic dominance; they differ in the interpretation of incomparabilities and ties.
The weak notion, which we will just call $\sd$-strategyproofness, prescribes that no agent can obtain an $\sd$-preferred outcome by lying about his preferences. 
Formally, an SDS $f$ is $\sd$-manipulable if there exist $R,R'\in\mathcal{R}^n$ and $i\in N$ such that $R_{-i} = R'_{-i}$ and $f(R') \sPref_i^{\sd} f(R)$. 
If an SDS is not $\sd$-manipulable, it is said to satisfy \emph{$\sd$-strategyproofness}. 

However, it may also be interpreted as a successful manipulation if an agent can obtain a lottery that is incomparable (according to stochastic dominance) to the lottery he obtains by reporting his preferences truthfully, since the former yields more expected utility than the latter for \emph{some} (rather than all) consistent utility functions. 
Strong $\sd$-strategyproofness requires that reporting one's preferences truthfully is a weakly dominant strategy.
Formally, an SDS $f$ satisfies \emph{strong $\sd$-strategyproofness} if $f(R) \pref_i^{\sd} f(R')$ for all $R,R'\in\mathcal{R}^n$ and $i\in N$ with $R_{-i} = R'_{-i}$.

It is a well known fact that $\rsd$ (and hence $\rd$) satisfies strong $\sd$-strategyproofness. For the domain of strict preferences, $\rd$ is the unique anonymous and \emph{ex post} efficient SDS that satisfies strong $\sd$-strategyproofness \citep{Gibb77a}. Within this domain $\rd$ is also $\sd$-efficient and hence also the unique anonymous SDS that satisfies $\sd$-efficiency and $\sd$-strategyproofness. More generally, it can be shown that every lottery that only randomizes over alternatives that are uniquely top ranked by some agent is $\sd$-efficient. However, $\rsd$ is not $\sd$-efficient on the full domain of preferences, which can be seen by again considering the example above.
It turns out that $\rsd(R) = \nicefrac{5}{12}\,a + \nicefrac{5}{12}\,b + \nicefrac{1}{12}\,c + \nicefrac{1}{12}\,d = p$. For $q = \nicefrac{1}{2}\,a + \nicefrac{1}{2}\, b$ we have $q\sPref_i^{\sd} p$ for all $i\in N$. Thus $\rsd$ is not $\sd$-efficient. In fact, every agent is \emph{strictly} better off in $q$ no matter what his utility function is (as long as it is consistent with his ordinal preferences). The failure of $\rsd$ to satisfy $\sd$-efficiency has been examined in great detail in the literature \citep[see, \eg][]{BMS05a,Mane08a,Mane09a,ChKo10a,BCKM12a,ABBH12a}.

\section{The Result}

We are now ready to show our main result, namely, that there exists no extension of $\rd$ to weak preferences that maintains its characteristic properties of efficiency and strategyproofness.

\begin{theorem}
\label{thm:main}
	There is no anonymous, neutral, $\sd$-efficient, and $\sd$-strategyproof extension of random dictatorship to the full domain of preferences when $m,n\ge 4$.
\end{theorem}

\begin{proof}
	We first prove that there is no SDS that satisfies the required properties for $n = 4$ and $m = 4$ and then use this statement to show that there is no such SDS for any larger number of agents and alternatives.
	 
	Without loss of generality, let $N = \{1,2,3,4\}$ and $A = \{a,b,c,d\}$ and assume for contradiction that $f$ is an SDS with the properties stated above. We will consider a sequence of preference profiles for which we (partially) determine the lottery returned by $f$. For a preference profile $R^k$ we denote by $p^k$ the lottery returned by $f$, \ie $p^k = f(R^k$). First, consider the following preference profile.
	\begin{align*}
		\pref^1_1\colon &\{a,c\},\{b,d\}&
		\pref^1_2\colon &\{b,d\},\{a,c\}&
		\pref^1_3\colon &\{a,d\},b,c&
		\pref^1_4\colon &\{b,c\},a,d
	\end{align*}
	Observe that ${\pref}_i = {\pref}_{\pi(i)}^{\sigma}$ for all $i\in N$ if $\pi = (1,2)(3,4)$ and $\sigma = (a,b)(c,d)$.
	Hence it follows from anonymity and neutrality that $f(R^1)(x) = f(R^1)(\sigma(x))$ for all $x\in A$ which implies that $p^1(a) = p^1(b)$ and $p^1(c) = p^1(d)$. If $p^1(c) = p^1(d) > 0$, then every agent $\sd$-prefers the lottery $\nicefrac{1}{2}\, a + \nicefrac{1}{2}\,b$ to $p^1$, contradicting $\sd$-efficiency. Hence $p^1(c) = p^1(d) = 0$ and it follows that $p^1 = \nicefrac{1}{2}\, a + \nicefrac{1}{2}\,b$. 
	\begin{align*}
		\pref^2_1\colon &\{a,c\},\{b,d\}&
		\pref^2_2\colon &\{b,d\},\{a,c\}&
		\pref^2_3\colon &a,d,\{b,c\}&
		\pref^2_4\colon &b,c,\{a,d\}
	\end{align*}
	With the same reasoning, we get that $p^2 = \nicefrac{1}{2}\, a + \nicefrac{1}{2}\,b$.
	
	We make another preliminary observation.
	\begin{align*}
		\pref^3_1\colon &a,c,\{b,d\}&
		\pref^3_2\colon &\{b,d\},a,c&
		\pref^3_3\colon &a,d,\{b,c\}&
		\pref^3_4\colon &\{b,c\},a,d
	\end{align*}
	With the permutations $\pi = (1,3)(2,4)$ and $\sigma = (a)(b)(c,d)$ it follows from anonymity and neutrality that $p^3(c) = p^3(d)$. 
	But no lottery with positive probability on both $c$ and $d$ is $\sd$-efficient for $R^3$. 
	Hence, $p^3(c) = p^3(d) = 0$. Assume for contradiction that $p^3(a) > \nicefrac{1}{2}$ and consider the following preference profile. 
	\begin{align*}
		\pref^4_1\colon &a,\{b,c,d\}&
		\pref^4_2\colon &\{b,d\},a,c&
		\pref^4_3\colon &a,d,\{b,c\}&
		\pref^4_4\colon &\{b,c\},a,d
	\end{align*}
	$\sd$-strategyproofness implies that $p^4(a) > \nicefrac{1}{2}$, as otherwise agent $1$ can benefit from reporting $\pref^3_1$ instead.  
	\begin{align*}
		\pref^5_1\colon &a,\{b,c,d\}&
		\pref^5_2\colon &\{b,d\},a,c&
		\pref^5_3\colon &a,\{b,c,d\}&
		\pref^5_4\colon &\{b,c\},a,d
	\end{align*}
	With the same reasoning as before but applied to agent $3$, we get $p^5(a) > \nicefrac{1}{2}$. Observe that $b$ Pareto-dominates $c$ and $d$ in $R^5$. Hence, $p^5(c) = p^5(d) = 0$ follows from $\sd$-efficiency. To derive a contradiction, we consider two more preference profiles.
	\begin{align*}
		\pref^6_1\colon &a,\{b,c,d\}&
		\pref^6_2\colon &b,\{a,c,d\}&
		\pref^6_3\colon &a,\{b,c,d\}&
		\pref^6_4\colon &\{b,c\},a,d
	\end{align*}
	Observe that again $p^6(c) = p^6(d) = 0$. If $p^6(a)\le\nicefrac{1}{2}$, agent $2$ in $R^5$ can benefit from reporting $\pref^6_2$ instead. Hence, $p^6(a) > \nicefrac{1}{2}$.
	Lastly, consider $R^7$.
	\begin{align*}
		\pref^7_1\colon &a,\{b,c,d\}&
		\pref^7_2\colon &b,\{a,c,d\}&
		\pref^7_3\colon &a,\{b,c,d\}&
		\pref^7_4\colon &b,\{a,c,d\}
	\end{align*}
	With the same reasoning as before but applied to agent $4$, $p^7(c) = p^7(d) = 0$ and $p^7(a) > \nicefrac{1}{2}$. However, it follows from anonymity and neutrality that $p^7(a) = p^7(b)$, a contradiction. Hence the assumption that $p^3(a) > \nicefrac{1}{2}$ was wrong. Combined with $p^3(c) = 0$, we get $p^3(a) + p^3(c)\le\nicefrac{1}{2}$.

	Now consider the preference profile $R^8$.
	\begin{align*}
		\pref^8_1\colon &\{a,c\},\{b,d\}&
		\pref^8_2\colon &\{b,d\},\{a,c\}&
		\pref^8_3\colon &a,d,\{b,c\}&
		\pref^8_4\colon &\{b,c\},a,d
	\end{align*}
	If agent $3$ reports $\pref^1_3$ instead, $f$ returns $p^1$. If $p^8(b) + p^8(c) > \nicefrac{1}{2}$, then $p^1(\sPref^8_3)^{\sd} p^8$, which contradicts $\sd$-strategyproofness. Hence, $p^8(b) + p^8(c) \le \nicefrac{1}{2}$.
	Similarly, if agent $4$ reports $\pref^2_4$ instead, $f$ returns $p^2$. 
	If $p^8(b) + p^8(c) < \nicefrac{1}{2}$, then $p^2(\sPref^8_4)^{\sd} p^8$, which again contradicts $\sd$-strategyproofness. Thus, together we have $p^8(b) + p^8(c) = \nicefrac{1}{2}$. 
	Moreover, if $p^8(d) > 0$ we necessarily have $p^2(\sPref^8_4)^{\sd} p^8$ given that $p^8(b) + p^8(c) = \nicefrac{1}{2}$. 
	Hence, $p^8(d) = 0$ and $p^8(a) = \nicefrac{1}{2}$. 
	\begin{align*}
		\pref^9_1\colon &a,c,\{b,d\}&
		\pref^9_2\colon &\{b,d\},\{a,c\}&
		\pref^9_3\colon &a,d,\{b,c\}&
		\pref^9_4\colon &\{b,c\},a,d
	\end{align*}
	If agent $2$ reports $\pref^3_2$ instead, then $f$ returns $p^3$. Assume for contradiction that $p^9(a) + p^9(c) > \nicefrac{1}{2}$. Then $p^3 (\sPref^9_2)^{\sd} p^9$, which contradictions $\sd$-strategyproofness. Hence, $p^9(a) + p^9(c)\le\nicefrac{1}{2}$. Moreover, if agent $1$ reports $\pref^8_1$, then $f$ returns $p^8$. Recall that $p^8(a) = \nicefrac{1}{2}$. If $p^9(a) < \nicefrac{1}{2}$, then together with $p^9(a) + p^9(c) \le\nicefrac{1}{2}$ this implies that $p^8(\sPref^9_1)^{\sd}p^9$, contradicting $\sd$-strategyproofness. So we get $p^9(a) = \nicefrac{1}{2}$. We use this insight to determine $p^8$. If $p^8(c) > 0$, then $p^8(\sPref^9_1)^{\sd}p^9$, which contradicts $\sd$-strategyproofness. Hence, $p^8(c) = 0$, which in turn implies that $p^8 = \nicefrac{1}{2}\,a + \nicefrac{1}{2}\, b$.
	\begin{align*}
		\pref^{10}_1\colon &\{a,c\},\{b,d\}&
		\pref^{10}_2\colon &b,\{a,c\},d&
		\pref^{10}_3\colon &a,d,\{b,c\}&
		\pref^{10}_4\colon &\{b,c\},a,d
	\end{align*}
	Note that $a$ Pareto-dominates $d$ in $R^{10}$, which implies that $p^{10}(d) = 0$ as $f$ is $\sd$-efficient. If agent $2$ reports $\pref^8_2$, then $f$ return $p^8$. If $p^{10}(b) > \nicefrac{1}{2}$, then $p^{10}(\sPref^8_2)^{\sd}p^8$, and if $p^{10}(b) < \nicefrac{1}{2}$, then $p^8(\sPref^{10}_2)^{\sd}p^{10}$. Both cases contradict $\sd$-strategyproofness. Hence, $p^{10}(b) = \nicefrac{1}{2}$. 
	\begin{align*}
		\pref^{11}_1\colon &c,a,\{b,d\}&
		\pref^{11}_2\colon &b,\{a,c\},d&
		\pref^{11}_3\colon &a,d,\{b,c\}&
		\pref^{11}_4\colon &\{b,c\},a,d
	\end{align*}
	Again, $d$ is Pareto-dominated by $a$ in $R^{11}$, and hence $p^{11}(d) = 0$. If agent $1$ reports $\pref^{10}_1$ instead, then $f$ returns $p^{10}$. If $p^{11}(b) < \nicefrac{1}{2}$, then $p^{11}(\sPref^{10}_1)^{\sd}p^{10}$, which contradicts $\sd$-strategyproofness. Hence, $p^{11}(b) \ge \nicefrac{1}{2}$.
	\begin{align*}
		\pref^{12}_1\colon &c,\{a,b\},d&
		\pref^{12}_2\colon &b,\{a,c\},d&
		\pref^{12}_3\colon &a,d,\{b,c\}&
		\pref^{12}_4\colon &\{b,c\},a,d
	\end{align*}
	Again, $d$ is Pareto-dominated by $a$ in $R^{12}$, and hence $p^{12}(d) = 0$. Moreover, with the permutations $\pi = (1,2)(3)(4)$ and $\sigma = (a)(b,c)(d)$ it follows from anonymity and neutrality that $p^{12}(b) = p^{12}(c)$. As $p^{11}(b) \ge \nicefrac{1}{2}$ and $p^{12}(b) = p^{12}(c)$, we have that $p^{12}(b)\le p^{11}(b)$. If $p^{12}(c) < p^{11}(c)$, then $p^{11}(\sPref^{12}_1)^{\sd}p^{12}$ and, on the other hand, if $p^{12}(c) > p^{11}(c)$, then $p^{12}(\sPref^{11}_1)^{\sd}p^{11}$. Both cases contradict $\sd$-strategyproofness. So together we have $p^{12}(c) = p^{11}(c)$. Next, if $p^{12}(a) > p^{11}(a)$, then agent $1$ in $R^{11}$ can benefit from reporting $\pref^{12}_1$ instead. So in summary, $p^{12}(a) + p^{12}(c) \le p^{11}(a) + p^{11}(c) \le\nicefrac{1}{2}$. As $p^{12}(b) = p^{12}(c)$, we have $p^{12} = \nicefrac{1}{2}\,b + \nicefrac{1}{2}\,c$.
	\begin{align*}
		\pref^{13}_1\colon &c,\{a,b\},d&
		\pref^{13}_2\colon &b,\{a,c\},d&
		\pref^{13}_3\colon &a,d,\{b,c\}&
		\pref^{13}_4\colon &b,c,a,d
	\end{align*}
	Recall that $f$ is an extension of $\rd$ and hence, $f(R^{13}) = \nicefrac{1}{4}\,a + \nicefrac{1}{2}\,b + \nicefrac{1}{4}\,c$. But $p^{12}(\sPref^{13}_4)^{\sd}p^{13}$, \ie agent $4$ can manipulate by reporting $\pref^{12}_4$ instead. This contradicts $\sd$-strategyproofness.
	
	Now let $|N|\ge 4$ and $|A|\ge 4$ be arbitrary and assume that $f$ is an anonymous, neutral, $\sd$-efficient, and $\sd$-strategyproof SDS. We use $f$ construct an SDS $f'$ that satisfies these properties for $N' = \{1,2,3,4\}$ and $A' = \{a,b,c,d\}$ which is a contradiction. Assume without loss of generality that $A'\subseteq A$. For every preference profile $R'$ on $N'$ and $A'$, choose some profile $R$ on $N$ and $A$ such that the preferences of the first $4$ agents over $A'$ coincide in $R$ and $R'$ and these agents prefer all alternatives in $A'$ to all alternatives in $A\setminus A'$ and the remaining agents are indifferent between all alternatives in $A$. Observe that only lotteries over $A'$ are $\sd$-efficient in $R$. Thus $f'(R') = f(R)$ is well-defined. It is easily verified that $f'$ inherits anonymity, neutrality, $\sd$-efficiency, and $\sd$-strategyproofness from $f$ which contradicts what we have shown above.
\end{proof}

\section*{Acknowledgments}
This material is based upon work supported by the Deutsche Forschungsgemeinschaft under grant {BR~2312/10-1}, by the TUM Institute for Advanced Study through a Hans Fischer Senior Fellowship, and a Stanford Graduate Fellowship.

\bibliographystyle{plainnat}

\end{document}